\documentclass[reqno]{amsart}

\usepackage{euler}
\usepackage{amsmath}
\usepackage{amsfonts}
\usepackage{amssymb}
\usepackage{verbatim}
\usepackage{amsthm}

\newtheorem{theorem}{Theorem}[section]
\newtheorem{lemma}[theorem]{Lemma}
\newtheorem{proposition}[theorem]{Proposition}
\newtheorem{corollary}[theorem]{Corollary}

\theoremstyle{definition}
\newtheorem{definition}[theorem]{Definition}
\theoremstyle{remark}
\newtheorem{remark}[theorem]{Remark}

\usepackage[pdftex]{hyperref}

\def\e{{\mathrm{e}}}

\def\Spec{{\mathrm{Spec}}}

\def\bA{{\mathbb A}}

\def\bC{{\mathbb C}}

\def\bH{{\mathbb H}}

\def\bM{{\mathbb M}}

\def\bP{{\mathbb P}}
\def\bQ{{\mathbb Q}}
\def\bR{{\mathbb R}}

\def\bZ{{\mathbb Z}}

\def\Q{{\mathbb Q}}
\def\Z{{\mathbb Z}}
\def\R{{\mathbb R}}

\def\cA{{\mathcal A}}
\def\cB{{\mathcal B}}

\def\cE{{\mathcal E}}
\def\cF{{\mathcal F}}

\def\cH{{\mathcal H}}

\def\cL{{\mathcal L}}

\def\cP{{\mathcal P}}

\def\cT{{\mathcal T}}

\def\cV{{\mathcal V}}

\def\cX{{\mathcal X}}
\def\cY{{\mathcal Y}}

\def\Aut{{\rm Aut}}

\def\d{{\rm d}}

\def\Gal{{\rm Gal}}
\def\GL{{\rm GL}}

\def\Hom{{\rm Hom}}

\def\Mod{{\rm Mod}}

\def\PGL{{\rm PGL}}
\def\PSL{{\rm PSL}}

\def\SL{{\rm SL}}
\def\Spec{{\rm Spec}}

\begin{document}

\title {Noncommutative geometry of groups like $\Gamma_{0}(N)$}
\author{Jorge Plazas}

\address{Concordia University - Department of Mathematics and Statistics - Montreal, QC, Canada}
\email{jorge.plazas@concordia.ca}

\subjclass[2010]{Primary 58B34; Secondary 11F06}

\keywords{Commensurability of $\bQ$-lattices, congruence subgroups}

\date{March 14, 2013 (Version v.1.0)}
\maketitle


\begin{abstract}
We show that the Connes-Marcolli $GL_{2}$-system can be represented on the Big Picture, a combinatorial gadget introduced by Conway in order to understand various results about congruence subgroups pictorially. In this representation  the time evolution of the system is implemented by Conway's distance between projective classes of commensurable lattices. We exploit these results in order to associate quantum statistical mechanical systems to congruence subgroups in a way appropriate for the study of monstrous moonshine.

\end{abstract}



\section{Introduction}
\label{section1}

\vspace{0.3 cm} \noindent The term \emph{monstrous moonshine} refers to the intriguing relation between sporadic simple groups and the arithmetic of automorphic functions.
In the late seventies observations of McKay and Thompson led Conway and Norton to the formulation of the monstrous moonshine conjectures \cite{ConwayNorton}. Given a congruence class $\langle m \rangle$ of elements in the monster group $\bM$ (the largest amongst the 26 sporadic simple groups) Thompson associates to it the generating series
\begin{eqnarray}
\label{hauptmoduli1}
  F_{\langle m \rangle}(q) &=& \sum_{j \in \bZ} h_{j}(m) \, q^{j} \\
   \nonumber 
  &=& q^{-1} +  h_{1}(m) q +  h_{2}(m) q^{2} + \dots
\end{eqnarray}
where the $h_{j}(m)$ are the character values at $m$ of particular representations of $\bM$ (its head representations). According to the main conjecture of Conway and Norton, proved by Borcherds in \cite{Borcherds}, to every $\langle m \rangle$ we can associate a congruence subgroup $\Gamma_{\langle m \rangle} \subset \GL_{2}^{+}(\bR)$ such that the field of $\Gamma_{\langle m \rangle}$-automorphic functions is generated by
\begin{eqnarray*}
f_{\langle m \rangle}(z)  &=& F_{\langle m \rangle}( \e^ {2\pi i z}).
\end{eqnarray*}

\vspace{0.3 cm} \noindent  Borcherds' proof uses techniques from string theory which enable him to show that the series $F_{\langle m \rangle}(q)$ satisfy enough identities to be uniquely determined by a finite number of its coefficients. Checking that each $F_{\langle m \rangle}(q)$ actually coincides with the corresponding automorphic function proposed by Conway and Norton completes the proof. A conceptual explanation of this fact is nevertheless still lacking and would constitute a great improvement of our understanding about sporadic groups and their relationship to automorphic functions.

\vspace{0.3 cm} \noindent The groups  $\Gamma_{\langle m \rangle}$ are commensurable with $\Gamma_{1} = \SL_{2}(\bZ)$ and as such, when viewed as groups of fractional linear transformations of the complex upper half plane $\mathfrak{H}$, they have the same set of cusps $\bQ\cup\{i \infty\}$. The moonshine theorem can then be restated by saying that for each $m\in\bM$ the Riemann surface $ \Gamma_{\langle m \rangle}   \diagdown \widetilde{\mathfrak{H}} =    \Gamma_{\langle m \rangle}   \diagdown  \left( \mathfrak{H} \cup \bQ\cup\{i \infty\} \right) $ is a genus zero Riemann surface whose field of functions is generated by $f_{\langle m \rangle}(z)$ i.e. $f_{\langle m \rangle}(z)$ is a \emph{principal modulus} for the \emph{genus zero} group $\Gamma_{\langle m \rangle}$.

\vspace{0.3 cm} \noindent Each of the groups $ \Gamma_{\langle m \rangle}$ occurring in moonshine contains some $\Gamma_{0}(N)$ to finite index. If we denote by $\overline{\Gamma}$ the image in $\PGL_{2}^{+}(\bR)$ of a subgroup $\Gamma \subset \GL_{2}^{+}(\bR)$ then  $\overline{ \Gamma}_{\langle m \rangle}$ is contained in the normalizer of $\overline{ \Gamma}_{0}(N)$ in $\PGL_{2}^{+}(\bR)$. We would like to understand better the automorphic function theory of these and similar groups in a way that might shed light into the monstrous moonshine phenomenon.

\vspace{0.3 cm} \noindent In his article \emph{Understanding groups like {$\Gamma_0(N)$}} \cite{Conway} Conway introduced a combinatorial gadget that simplifies the analysis of groups related to $\Gamma_{0}(N)$ making various facts about their structure visually apparent. Conway's main idea is to view $\PSL_{2}(\bZ)$ as the automorphism group of a fixed lattice $\Lambda_{1} \subset \bC$ and to study groups commensurable with $\PSL_{2}(\bZ)$ in terms of lattices commensurable with $\Lambda_{1}$. The big picture $\cB$ is a graph whose vertices are given by projective classes of lattices commensurable with $\Lambda_{1}$ whose edge structure is defined in terms of the natural distance between these classes. Using the big picture Conway gives beautiful simple proofs of the index formula for Hecke operators, the Atkin-Lehner theorem describing the normalizer $N_{\PSL_{2}(\bR)}(\overline{ \Gamma}_{0}(N))$ and Helling's theorem on maximal subgroups of $\PSL_{2}(\bC)$ commensurable with $\PSL_{2}(\bZ)$.

\vspace{0.3 cm} \noindent Considering lattices, together with a labeling of the torsion points in the corresponding elliptic curve, leads to the notion of $\bQ$-lattices introduced by Connes and Marcolli in \cite{ConnesMarcolli1} where spaces of commensurability classes of these objects are studied. From the point of view of topology the space of commensurability classes of $\bQ$-lattices is a degenerate quotient,  from the standpoint of noncommutative geometry however this quotient makes perfect sense and defines a noncommutative space with a rich structure reflecting the arithmetic of modular functions. The algebra of continuous functions on the noncommutative space of commensurability classes of $\bQ$-lattices is by definition the $C^{*}$-algebra $\cA$ corresponding to the equivalence relation given by commensurability. The $C^{*}$-algebra $\cA$ comes equipped with a natural uniparametric group of automorphisms $\sigma_{t}$ defining a quantum statistical mechanical system. Connes and Marcolli show that the symmetries and equilibrium states of the system $(\cA, \sigma_{t})$ encode the arithmetic of abelian extensions of the modular field $\cF$ (the field of modular functions of arbitrary level).

\vspace{0.3 cm} \noindent In this article we show that the Connes-Marcolli system can be represented on the big picture $\cB$. Using this representation we interpret Conway's techniques in the context of algebras of operators. Each congruence subgroup $\Gamma$ of the type mentioned above singles out an algebra of arithmetic observables of the Connes-Marcolli system. We show that the arithmetic of automorphic functions of groups related to $\Gamma_{0}(N)$ is encoded in terms of the system $(\cA, \sigma_{t})$ via this algebra. Our main motivation comes from potential applications to monstrous moonshine.

\vspace{0.3 cm} \noindent I thank John McKay for many useful discussions leading to this work. The present work has been supported by the Granada Excellence Network of Innovation Laboratories GENIL.


\vspace{0.3 cm}
\section{Commensurability classes of lattices and Conway's big picture}
\label{Section2}

\vspace{0.3 cm} \noindent  In this section we recall the definition of the big picture introduced by Conway in \cite{Conway}. We also look at the local structure of the big picture and relate it to the Bruhat-Tits buildings of $\GL_{2}$ over the local fields $\bQ_{p}$.

\subsection{Generalities}
\label{Section2.1}

\vspace{0.2 cm} \noindent By a \emph{lattice} we mean a lattice in $\bC$ i.e. a discrete co-compact subgroup $\Lambda$ of $\bC$. We recall the following definitions:

\begin{itemize}
  \item Two lattices $\Lambda$ and $\Lambda'$ are \emph{commensurable} if they are commensurable as groups i.e. if their intersection $\Lambda\cap \Lambda'$ has finite rank in both $\Lambda$ and $\Lambda'$.
  \item Two lattices $\Lambda$ and $\Lambda'$ are \emph{projectively equivalent} if there exists a scalar $\alpha \in \bQ^{*}$ such that $\Lambda = \alpha \Lambda'$.
  \item Two lattices $\Lambda$ and $\Lambda'$ are \emph{homothetic} or \emph{equivalent up to scale} if there exists a scalar $ \lambda \in \bC^{*}$ such that $\Lambda = \lambda \Lambda'$.
\end{itemize}

\vspace{0.3 cm} \noindent  Thus projectively equivalent lattices are homothetic and any lattice $\Lambda$ is equivalent up to scale to a lattice of the form $\Lambda_{z} = z\bZ+\bZ$ for a complex number $z$ in the upper half-plane $\mathfrak{H}$. When passing to homothety classes the $GL_{2}^{+}(\bR)$-action on lattices corresponds to the action of $GL_{2}^{+}(\bR)$ on $\mathfrak{H}$ by fractional linear transformations
\begin{eqnarray*}
g z &=  & \frac{az +b}{cz + d},  \qquad  g = \left(\begin{array}{cc} a & b\\c & d \end{array}\right)  \in GL_{2}^{+}(\bR), \, z\in  \mathfrak{H}
\end{eqnarray*}
We identify  $PGL_{2}^{+}(\bR) \simeq PSL_{2}(\bR)$ with the group of automorphisms of $\mathfrak{H}$ via this action. Under this identification the set of lattices up to scale corresponds to the quotient $\PSL_{2}(\bZ)  \diagdown   \mathfrak{H}$.

\vspace{0.3 cm} \noindent Note also that if $\Lambda$ and $\Lambda'$ are commensurable lattices then they are homothetic if and only if they are projectively equivalent and that commensurability is well defined at the level of projective classes of lattices. We denote by $\cL$ the set of all lattices in $\bC$. Given a lattice $\Lambda_{1}$ we denote by $\cL_{\Lambda_{1}}$ the set of all lattices commensurable with $\Lambda_{1}$ and by $\cP\cL_{\Lambda_{1}}$ the set of all projective classes of lattices commensurable with the class of $\Lambda_{1}$. By the previous remark  $\cP\cL_{\Lambda_{1}}$ coincides with  the set of all homothety classes of lattices commensurable with the class of $\Lambda_{1}$. 

\vspace{0.3 cm} 

\subsection{Conway's hyperdistance and the big picture} 
\label{Section2.2}

 Following Conway  \cite{Conway} we introduce a metric on the set $\cP\cL_{\Lambda_{1}}$:

\begin{definition}
Given two commensurable lattices $\Lambda$ and $\Lambda'$ the \emph{hyperdistance} between them is the smallest positive integer $N$ for which there exist scalars $\alpha, \alpha' \in \bQ^{*}$ such that
\begin{eqnarray*}
 \left[ \Lambda : \alpha' \Lambda' \right]  &= \,  N  \,  = & \left[ \Lambda' : \alpha \Lambda \right]
\end{eqnarray*}
\end{definition}

\vspace{0.3 cm} \noindent  We denote the hyperdistance between two lattices $\Lambda$ and $\Lambda'$ by $\delta(\Lambda, \Lambda')$. The function $\delta$ is symmetric and well defined at the level of projective classes. Two commensurable lattices $\Lambda$ and $\Lambda'$ are projectively equivalent if and only if $\delta(\Lambda, \Lambda') = 1$. The logarithm of the hyperdistance
\begin{eqnarray*}
\d (\Lambda, \Lambda') &= & \log \left( \delta(\Lambda, \Lambda') \right)
\end{eqnarray*}
defines a metric on the set $\cP\cL_{\Lambda_{1}}$ of projective classes of lattices commensurable to a fixed lattice $\Lambda_{1}$.

\vspace{0.3 cm} \noindent  Two lattices $\Lambda$ and $\Lambda'$ are commensurable if and only if $\bQ \Lambda = \bQ \Lambda'$. Once we have chosen a basis
$\{\omega_{1}, \omega_{2} \}$ for a lattice $\Lambda_{1}$ we can express commensurability in terms of elements of $GL_{2}^{+}(\bQ)$. Identifying $PSL_{2}(\bZ)$ with the automorphism group of a lattice we have the following:

\begin{remark}
\label{basisrmrk}
Given a basis $\{\omega_{1}, \omega_{2} \}$ for a  lattice $\Lambda_{1}$ we can identify the  quotient $ SL_{2}(\bZ)   \diagdown   GL_{2}^{+}(\bQ)$ (resp. $PSL_{2}(\bZ) \diagdown  \PGL_{2}^{+}(\bQ)$) with the set $\cL_{\Lambda_{1}}$ (resp. $\cP\cL_{\Lambda_{1}}$) via
 \begin{eqnarray*}
\left(\begin{array}{cc} a & b \\c & d\end{array}\right) & \longmapsto &  (a \omega_{1} + b  \omega_{2})\Z  + (c \omega_{1} + d  \omega_{2}) \Z
\end{eqnarray*}
\end{remark}

\vspace{0.3 cm} \noindent  In what follows it will be convenient to interpret the hyperdistance in terms of matrices (see \cite{Duncan} for a similar approach). Given $g \in GL_{2}^{+}(\bQ)$ we let $\alpha_{g}$ be the smallest positive rational number such that
 \begin{eqnarray*}
 \alpha_{g} g \in M_{2}(\bZ)
\end{eqnarray*}
Define now:
 \begin{eqnarray*}
\delta_{1}(g) &=& \det (  \alpha_{g} g )\\
&=&  \alpha_{g}^{2}  \det ( g )
 \end{eqnarray*}

\vspace{0.3 cm} \noindent The main properties of the map $\delta_{1}$ are summarized in the following
\begin{lemma}
Let $g \in GL_{2}^{+}(\bQ)$ , then
\begin{itemize}
  \item $\delta_{1} (\alpha g) =\delta_{1} (g)$ for all $\alpha \in  \bQ^{*}$
  \item $\delta_{1}(\gamma g) =\delta_{1} (g) =  \delta_{1}(g \gamma)$ for all $\gamma \in SL_{2}(\bZ)$
\end{itemize}
In particular $\delta_{1}$ descends to a map:
 \begin{eqnarray*}
\delta_{1} : \;    \PSL_{2}(\bZ)  \diagdown   \PGL_{2}^{+}(\bQ)& \longrightarrow &\bZ_{> 0}
 \end{eqnarray*}
\end{lemma}

\vspace{0.3 cm} \noindent If we denote by $ \bar{g}  $ the class in $\PGL_{2}^{+}(\bQ) $ of an element $g \in \GL_{2}^{+}(\bQ)$
 then we can express the map $\delta_{1}$ as:
 \begin{eqnarray*}
\delta_{1}(\bar{g} ) &=& \det (s (\bar{g}) )
 \end{eqnarray*}
where
 \begin{eqnarray}
 \label{eqsection}
 s : PGL_{2}^{+}(\bQ) & \longrightarrow & M_{2}^{+}(\bZ) \\ \nonumber
  \bar{g}  & \longmapsto & \alpha_{g} g
 \end{eqnarray}
is the section of $\PGL_{2}^{+}(\bQ) =    \bQ^{*} \diagdown  \GL_{2}^{+}(\bQ) $ with image in $ \mathrm{M}_{2}^{+}(\bZ) = \GL_{2}^{+}(\bQ)\cap  M_{2}(\bZ)$ obtained by sending an element of $\PGL_{2}^{+}(\bQ)$ to its representative with integer coefficients in its lowest terms.
We will denote the map obtained by passing to quotients by the same letter:
 \begin{eqnarray*}
 s : \;      \PSL_{2}(\bZ)  \diagdown   \PGL_{2}^{+}(\bQ) & \longrightarrow & \SL_{2}(\bZ)  \diagdown    M_{2}^{+}(\bZ) \, .
 \end{eqnarray*}

\vspace{0.3 cm} \noindent Taking into account Remark~\ref{basisrmrk} the map $\delta_{1}$ is an integer valued function on the set of all lattices commensurable with a given lattice:

\begin{lemma}
\label{hypermetric}
Fix a lattice $\Lambda_{1}$ in $\bC$. Choosing a basis of $\Lambda_{1}$ identify $\cP\cL_{\Lambda_{1}}$ with $ \PSL_{2}(\bZ)  \diagdown   \PGL_{2}^{+}(\bQ)$. Given two lattices $\Lambda$ and $\Lambda'$ corresponding to matrices $ g $ and $h $ their hyperdistance is given by:
 \begin{eqnarray*}
\delta(  \Lambda  ,  \Lambda'  ) &=& \delta_{1}( g h^{-1}).
 \end{eqnarray*}
Taking $\d_{1} = \log\delta_{1}$  we can express the distance between commensurable projective classes of lattices as
 \begin{eqnarray*}
\d :   \cP\cL_{\Lambda_{1}} \times \cP\cL_{\Lambda_{1}}  & \longrightarrow & \R _{\geq 0} \\
(\nu \, , \,  \nu' )\,  & \longrightarrow &  \d_{1} ( g h^{-1}).
 \end{eqnarray*}
where $\nu$ and $\nu'$ denote respectively the projective classes of the lattices $\Lambda$ and $\Lambda'$.

\end{lemma}

\begin{definition}
Fix a lattice $\Lambda_{1}$ in $\bC$. The \emph{Big Picture} is the graph $\cB$ whose set of vertices
 \begin{eqnarray*}
\cV_{\cB} &=&  \cP\cL_{\Lambda_{1}}
 \end{eqnarray*}
corresponds to the set of projective classes of lattices commensurable with $\Lambda_{1}$ with two vertices $\nu$ and $\nu'$ being joined by an edge if and only if the hyperdistance $\delta(  \Lambda  ,  \Lambda')$  between the corresponding lattices is a prime number.
\end{definition}


\subsection{Local structure of the big picture} 
\label{Section2.3}

\vspace{0.3 cm} \noindent For each prime number $p$ let $\cT_{p}$ be the subgraph of $\cB$ whose vertices are given by projective classes of lattices whose hyperdistance to
$\Lambda_{1}$ is a power of $p$. We have the following decomposition of  $\cB$ in terms of the  $\cT_{p}$ (cf. \cite{Conway})

\begin{lemma}
\begin{enumerate}
  \item For each prime number $p$ the graph  $\cT_{p}$ is a $(p+1)$-valent tree.
  \item The big picture $\cB$ factorizes as
   	\begin{eqnarray*}
	\cB &=& \prod_{p}  \cT_{p}
	\end{eqnarray*}
  where the product runs over all primes.
\end{enumerate}
\end{lemma}

\begin{proof}
\vspace{0.3 cm} \noindent   To see this note that if a vertex $\nu$ corresponds to the projective class of a lattice $\Lambda$ with basis $\{ \omega_{1}, \omega_{2} \}$  then any lattice at hyperdistance $p$ from $\Lambda$ is projectively equivalent to a sublattice $\Lambda' \subset \Lambda$ with $\left[ \Lambda : \Lambda' \right]=p$ and the set of such lattices correspond under the identification in Remark~\ref{basisrmrk} to the p+1 matrices
 \begin{eqnarray*}
\left(\begin{array}{cc}p & 0 \\0 & 1\end{array}\right)  &\text{ and }& \left(\begin{array}{cc}1 & b \\0 &  p \end{array}\right)   \, \quad 0\leq b \leq p-1
\end{eqnarray*}
so given a prime number $p$ each vertex of $\cB$ is connected with exactly $p+1$ vertices. Also note that given two vertices $\nu$ and $\nu'$ with $\delta(  \nu ,  \nu') = p^{n}$ there is a unique path of length $n$ between them and this in turn is the shortest path connecting $\nu$ with $\nu'$. The second assertion follows from the definition of $\cB$.
\end{proof}

\vspace{0.4 cm} \noindent As a topological space the big picture $\cB$ thus coincides with the restricted product of the one-dimensional simplicial complexes $ \cT_{p}$ with respect to the base point $\nu_{1}$, the projective class of $\Lambda_{1}$. We clarify now the arithmetic meaning of this factorization.  

\vspace{0.4 cm} \noindent  For a prime number $p$ we call a lattice in $V_{p} = \Q_{p}^{2}$ a \emph{local lattice}, so a local lattice is by definition a finitely generated $\Z_{p}$-submodule which generates $V_{p}$ as a vector space over $\Q_{p}$ (in particular a lattice is free of rank 2 as a $\Z_{p}$-module). Two local lattices are projectively equivalent if they differ by multiplication by an element in $\bQ_{p}^{*}$. We say that two projective classes of local lattices are \emph{adjacent} if we can choose representatives $\Lambda_{p}$ and   $\Lambda'_{p} $ with $\Lambda_{p}' \subset \Lambda_{p} $ and $\Lambda_{p}/ \Lambda_{p}' \simeq \Z / p\Z$. The Bruhat-Tits building of $\GL_{2}( \Q_{p})$ is by definition the graph whose vertices are given by projective classes of local lattices with an edge connecting two vertices if and only if the corresponding projective classes of lattices are adjacent (cf. \cite{Serre}). We have the following: 

\begin{proposition}
For any prime number $p$ the tree $ \cT_{p}$ is canonically isomorphic to the Bruhat-Tits building of $\GL_{2}( \Q_{p})$.
\end{proposition}
\begin{proof}
Fix a lattice $\Lambda_{1}$ in $\bC$ and let $V= \bQ \Lambda_{1}$. Then $\bQ \Lambda = V$ for any lattice commensurable with $\Lambda_{1}$ and
$V\otimes_{\bQ}\bQ_{p}$ is isomorphic as a $\bQ_{p}$-vector space to $V_{p}$. Under this isomorphism the $\bZ_{p}$-module $ \Lambda^{(p)}$ generated by the image in $V\otimes_{\bQ}\bQ_{p}$ of a lattice $\Lambda \in \cL_{\Lambda_{1}}$ is a local lattice. The local lattice $\Lambda^{(p)}$ associated to $\Lambda$ is different  from $ \Lambda_{1}^{(p)}$ if and only if $p$ divides  $\delta (\Lambda, \Lambda_{1}) $ in which case the corresponding vertex in the Big Picture $\nu\in \cV_{\cB}$ belongs to $ \cT_{p}$ and is connected to the vertex $\nu_{1}$ corresponding to the projective class of $\Lambda_{1}$ by a path of length $u_{p} (\delta (\Lambda, \Lambda_{1}))$ where $u_{p} $ is the $p$-adic valuation given by $|r|_{p} = p^{-u_{p}(r)}$. Two classes of lattices $\nu$ and $\nu'$ in $\cT_{p}$ correspond to adjacent lattices in the Bruhat-Tits building of $\GL_{2}( \Q_{p})$ if and only the hyperdistance between two corresponding representatives $\Lambda$ and $\Lambda'$ is p, that is if and only if $\nu$ and $\nu'$ are connected by an edge in $ \cT_{p}$.   
\end{proof}

\begin{corollary}
\label{collorary1}
The big picture is isomorphic to the product over all prime numbers $p$ of the Bruhat-Tits buildings of $\GL_{2}( \Q_{p})$.
\end{corollary}

\vspace{0.3 cm} 


\subsection{The big picture and congruence subgroups} 
\label{Section2.4}
 Given a subgroup $\Gamma \subset \GL_{2}^{+}(\bQ)$ we will denote by $\overline{\Gamma}$ its image in $\PGL_{2}^{+}(\bQ)$. Likewise we will denote subgroups of $PGL_{2}^{+}(\bQ)$ with a bar. In particular for a positive integer $N$ we denote by $\overline{\Gamma}(N)$ the image in $\PGL_{2}^{+}(\bQ)$ of the principal congruence subgroup
 \begin{eqnarray*}
\Gamma(N) &=& \{ \gamma  \, = \,   \left(\begin{array}{cc} a & b \\c & d\end{array}\right) \in \SL_{2}(\bZ)  \, | \,     a\equiv d \equiv   1 \,  \Mod  N , \,  b \equiv c \equiv 0 \, \Mod N  \}
 \end{eqnarray*}
and by $\overline{\Gamma}_{0}(N)$ the image in $\PGL_{2}^{+}(\bQ)$ of
 \begin{eqnarray*}
\Gamma_{0}(N) &=& \{   \gamma  \, = \,  \left(\begin{array}{cc} a & b \\c & d\end{array}\right)  \SL_{2}(\bZ)  \, |  \, c \equiv 0  \,  \Mod N  \}
 \end{eqnarray*}
We say that $\Gamma \subset \GL_{2}^{+}(\bQ)$ (resp. $\overline{\Gamma} \subset \PGL_{2}^{+}(\bQ)$) is a \emph{congruence subgroup} if it contains some $\Gamma(N)$ (resp. $\overline{\Gamma}(N)$) with finite index. 

\vspace{0.3 cm} \noindent Looking at the action of elements of $\PGL_{2}^{+}(\bQ)$ on the Big Picture we can understand various results about congruence subgroups. In what follows we describe briefly Conway's approach where congruence subgroups and their supergroups are viewed as pointwise and setwise stabilizers of particular finite subgraphs of $\cB$ \footnote{Our notation for various groups and lattices differs slightly from that used by Conway}.

\vspace{0.3 cm} \noindent Fix a lattice $\Lambda_{1}$ together with a basis $\{ \omega_{1}, \omega_{2} \}$. We let $\Lambda_{N}$ be the lattice with basis $\{ N \omega_{1}, \omega_{2} \}$ and denote by $\nu_{1}$ and $\nu_{N}$ the projective classes of $\Lambda_{1}$ and $\Lambda_{N}$. We have then $\delta(\Lambda_{1}, \Lambda_{N}) =\delta( \nu_{1}, \nu_{N}) = N$ and for any two lattices $\Lambda$ and $\Lambda'$ at hyperdistance $\delta(\Lambda, \Lambda') = N$ we can find a matrix $g \in \GL_{2}^{+}(\bQ)$ such that $g \Lambda = \Lambda_{1}$ and $ g \Lambda'= \Lambda_{N}$ and so $\bar{g} \in \PGL_{2}^{+}(\bQ)$ takes the corresponding projective classes $\nu$ and $\nu'$ to $\nu_{1}$ and $\nu_{N}$.

\vspace{0.3 cm} \noindent  Conway observes that the subgroup $\overline{\Gamma}_{0}(N|1) := \overline{\Gamma}_{0}(N)$ is the joint stabilizer of $\nu_{1}$ and $\nu_{N}$ and defines for  two lattices $\Lambda$ and $\Lambda'$ the group $\overline{\Gamma}_{0}(\nu | \nu')$ to be the joint stabilizer of their projective classes. In the case of the lattices $\Lambda_{l}$ and $\Lambda_{M}$ corresponding to a positive integer $M$ and a divisor $l$ of $M$ we write $\overline{\Gamma}_{0}(M|l) =  \overline{\Gamma}_{0}(\nu_{M}|\nu_{l})$ and note that this is the conjugate of $\overline{\Gamma}_{0}(N)$, $N=\frac{M}{l}$, by the matrix  $g_{l} = \left(\begin{array}{cc} l & 0 \\ 0 & 1 \end{array}\right)$ taking $\nu_{1}$ to $\nu_{l}$ and $\nu_{N}$ to $\nu_{M}$. 

\vspace{0.3 cm} \noindent In order to describe the normalizer of $\overline{\Gamma}_{0}(N)$ it is useful to consider the subgraph lying between $\nu_{1}$ and $\nu_{N}$. First define for a prime power $p^n$ the $(1|p^n)$-\emph{thread }$\mathfrak{t}_{p^n}$ to be the unique shortest path between $ \nu_{1}$ and $\nu_{p^n}$ in $ \cT_{p}$. For a positive integer $N$ the $(1|N)$-thread $\mathfrak{t}_{N}$ is the product of the threads $\mathfrak{t}_{p_i^{n_{i}}}$ corresponding to the prime decomposition of $N$. The $(N|1)$-thread is fixed pointwise by $\overline{\Gamma}_{0}(N)$. Symmetries of this thread correspond then to cosets of  $\overline{\Gamma}_{0}(N)$, these are the Atkin-Lehner involutions $\{W_{e} | e \text{ exact divisor of } N \}$ introduced by  Atkin and Lehner in \cite{AtkinLehner}. The extension of $\overline{\Gamma}_{0}(N)$ by its Atkin-Lehner involutions is denoted by $\overline{\Gamma}_{0}(N)+$ and is the setwise stabilizer of the $(1|N)$-thread. Likewise, for $lN =M$  we denote by $\overline{\Gamma}_{0}(M|l)+$ the conjugate of $\overline{\Gamma}_{0}(N)+$ by $g_{l}$. The set of projective classes of lattices fixed by  $\overline{\Gamma}_{0}(N)$ constitutes what Conway calls the $(N|1)$-\emph{snake} $\mathfrak{s}_{N}$ and is given by the classes of lattices whose hyperdistance to $\mathfrak{t}_{N}$ divides 24. Looking at the setwise stabilizer of $\mathfrak{s}_{N}$ one obtains the Atkin-Lehner theorem: the normalizer of $\overline{\Gamma}_{0}(N)$ in $\PSL_{2}(\bR)$ is the group $\overline{\Gamma}_{0}(\frac{N}{h}| h)+$ where $h$ is the largest divisor of $24$ for which $h^{2}$ divides $N$.


\section{Commensurability classes of $\bQ$-lattices and the Connes-Marcolli $\GL_{2}$-system}
\label{section3}

\vspace{0.3 cm} \noindent  In this section we recall the definition of the $GL_{2}$-system introduced in \cite{ConnesMarcolli1}.

\begin{definition}
A 2-dimensional $\bQ $-lattice is given by a pair $( \Lambda ,   \phi )$ where $\Lambda $ is a lattice in $\bC $ and
 \begin{eqnarray*}
 \phi :  \bQ ^{2} / \bZ  ^{2}  \rightarrow  \bQ \Lambda / \Lambda
 \end{eqnarray*}
 is a group homomorphism.
\end{definition}

\vspace{0.3 cm} \noindent  A $\bQ $-lattice $( \Lambda ,   \phi )$ is \emph{invertible} if  $ \phi$ is an isomorphism. Two $\bQ $-lattices $( \Lambda ,   \phi)$ and $( \Lambda' ,   \phi')$ are \emph{commensurable} if $\Lambda$ and  $\Lambda'$ are commensurable as lattices (thus $\bQ \Lambda = \bQ \Lambda'$) and
 \begin{eqnarray*}
  \phi =  \phi'  \mod  (\Lambda + \Lambda').
 \end{eqnarray*}

 \begin{lemma}[\cite{ConnesMarcolli1}]
The relation of commensurability between $\bQ $-lattices is an equivalence relation.
\end{lemma}

\vspace{0.3 cm} \noindent  We say that two $\bQ $-lattices $( \Lambda ,   \phi )$ and $( \Lambda' ,   \phi' )$ are \emph{equivalent up to scale} if there exist $ \lambda \in \bC^{*}$ such that $(  \lambda \Lambda ,  \lambda  \phi ) = ( \Lambda' ,   \phi')$. The $\bC^{*}$-action on  $\bQ $-lattices preserves the commensurability relation.  Denote by $\widehat{ \bZ } =  \lim_{ \overleftarrow{n} }    \bZ / n\bZ $ the ring of profinite integers. Identifying $ \Hom (\bQ ^{2} / \bZ  ^{2}) $ with  $ M_{2}(\widehat{ \bZ } )$ we have the following:


\begin{remark}
The set of $\bQ $-lattices up to scale can be identified with the product $ M_{2}(\widehat{ \bZ } ) \times \mathfrak{H} $ modulo $SL_{2}(\bZ) $ where the action is given by left multiplication on the first factor and by fractional linear transformations on the second factor.
\end{remark}

 \vspace{0.3 cm} \noindent  As above once we choose a basis for a lattice, commensurability can be expressed in terms of elements of $GL_{2}^{+}(\bQ)$. Two $\bQ $-lattices, corresponding respectively to pairs $( \rho_{1}, z_{1} )$ and $( \rho_{2}, z_{2} )$ in  $M_{2}(\widehat{ \bZ } ) \times \mathfrak{H}$, are commensurable if and only if there exists an element  $g \in GL_{2}^{+}(\bQ)$ such that $\rho_{2} = g \rho_{1} $ and $z_{2} = g z_{1}$.

 \vspace{0.3 cm} \noindent  We turn now to the arithmetic geometry of the space of commensurability classes of $\bQ $-lattices. Following  \cite{ConnesMarcolli1} we study  this space from the vantage point of noncommutative geometry. We start by introducing a noncommutative algebra of functions naturally encoding the commensurability relation of $\bQ $-lattices up to scale.

  \vspace{0.3 cm} \noindent
Let $\Gamma_{1} =SL_{2}(\bZ)$ and let $\cA_{c}$ be the set of $\bC$-valued functions with compact support on
 \begin{eqnarray*}
\cY &=& \{  (g,  \rho, z) \in GL_{2}^{+}(\bQ) \times M_{2}(\widehat{ \bZ }) \times \mathbb{H} \; | \; g  \rho \in M_{2}(\widehat{ \bZ }) \}
\end{eqnarray*}
such that
\begin{eqnarray*}
f(  \gamma g,  \rho, z) = f(g, \rho,  z) &\text{   and    }&  f(g  \gamma ,  \rho, z) = f(g, \gamma  \rho,  \gamma z)
\end{eqnarray*}
for all $\gamma  \in \Gamma_{1}$.

  \vspace{0.3 cm} \noindent
The product
 \begin{eqnarray}
 \label{product}
f_{1} \ast f_{2}  (g,  \rho, z) &=& \sum_{  s \in  \Gamma_{1}  \diagdown    GL_{2}^{+}(\bQ) \, : \,  s \rho \in M_{2}(\widehat{ \bZ })  }
f_{1}  (g s^{-1},  s \rho,  s z)  f_{2}  (s ,  \rho, z)
 \end{eqnarray}
and convolution
  \begin{eqnarray*}
f^{*} (g,  \rho, z) &=&\overline{ f(g^{-1}, g \rho, g z) }.
 \end{eqnarray*}
 define a $*$-algebra algebra structure on $\cA_{c}$.

\vspace{0.3 cm} \noindent   Given any element $y  =(\rho, z) \in M_{2}(\widehat{ \bZ } ) \times \mathfrak{H}$ let
 \begin{eqnarray*}
\cY_{y }&=& \{ g  \in GL_{2}^{+}(\bQ)  \; | \; g  \rho \in M_{2}(\widehat{ \bZ }) \}.
 \end{eqnarray*}
The algebra  $\cA_{c}$  admits a representation as an algebra of bounded operators on $\cH_{y}= \ell^{2}(\Gamma  \, \diagdown   \cY_{y})$ given by:
 \begin{eqnarray}
  \label{rep1}
\pi_{y} (f) \xi (g) &=& \sum_{  h \in \Gamma_{1}  \diagdown   \cY_{y}  } f  (g h^{-1},  h \rho,  h z)  \xi  (h).
\end{eqnarray}

\begin{definition}
The algebra of \emph{continuous functions on the space of commensurability classes of $\bQ$-lattices up to scale} is the $C^{*}$-algebra $\cA$ obtained as the completion of $\cA_{c}$ under the norm
 \begin{eqnarray*}
\|  f  \| &=& \sup_{y} \| \pi_{y} (f)  \|
 \end{eqnarray*}
\end{definition}

\vspace{0.3 cm} \noindent  The algebra $\cA $ admits a norm-continuous uniparametric group of automorphisms given by:
 \begin{eqnarray*}
(\sigma_{t} f )(g,  \rho, z) &=& \det (g)^{i t } \,  f (g,  \rho, z), \qquad t \in \bR
 \end{eqnarray*}

 \vspace{0.3 cm} \noindent  A pair consisting of a $C^{*}$-algebra together with a uniparametric group of automorphisms defines a \emph{quantum statistical mechanical system}. We refer the reader to \cite{Bratteli} for various generalities about quantum statistical mechanical systems and to \cite{ConnesMarcolli1} and references therein for examples of these systems arising in arithmetic noncommutative geometry.

 \begin{definition}
We call the pair $(\cA, \sigma_{t} )$ with $\cA$ and $\sigma_{t}$ as above, the \emph{$GL_{2}$-system}. We will refer to $\cA$ as the \emph{algebra of observables} of the system and to  $\{ \sigma_{t} \in \Aut (\cA) \, | \,     t \in \bR \}$ as the \emph{time evolution of the system}.
\end{definition}

 \vspace{0.3 cm} \noindent Our interest in the $GL_{2}$-system comes from the fact that its thermodynamic properties encode the arithmetic theory of modular functions to an extend which makes it possible for us to capture aspects of moonshine theory.

 \vspace{0.3 cm} 

\section{From Conway's hyperdistance to the time evolution}
\label{section4}

 \vspace{0.3 cm} \noindent  In this section we relate the time evolution of the $GL_{2}$-system to Conway's hyper distance, this is done by showing that the algebra of observables can be represented as an algebra of operators acting on functions on the set of vertices of $\cB$. Before doing this however we make a few remarks on the extension of Conways's hyperdistance from the setting of commensurability classes of lattices to that of commensurability classes of $\bQ$-lattices. 
 
\subsection{Hyperdistance for $\bQ$-lattices} 
For $\bQ$-lattices Conway's definition of the hyperdistance continues to make sence but one has to take into account the divisibility properties of the labeling of torsion points. With the notation being as in Section~\ref{Section2.2} we have that for $\rho \in M_{2}(\widehat{ \bZ } )$ and $g \in GL_{2}^{+}(\bQ)$ 
 
  \begin{eqnarray*}
g \rho \in M_{2}(\widehat{ \bZ } ) & \text{ if and only if }& \frac{1}{ \alpha_{g}} \rho \in M_{2}(\widehat{ \bZ } )
 \end{eqnarray*}

 \vspace{0.3 cm} \noindent  In particular for $g \in GL_{2}^{+}(\bQ)$ with  $\alpha_{g} = N $ we have that $g \rho \in M_{2}(\widehat{ \bZ } ) $ only when the matrix elements of $ \rho$ are  $N$-divisible.

 \vspace{0.3 cm} \noindent Given a $\bQ$-lattice $( \Lambda_{1} ,  \phi_{1} )$ represented up to scale by $y_{1}  = (\rho_{1}, z_{1}) \in M_{2}(\widehat{ \bZ } ) \times \bH$, we let $\cL_{( \Lambda_{1} ,   \phi_{1} ) }^{\bQ} = \cL_{y_{1}}^{\bQ}  $ be the set of all $\bQ$-lattices commensurable with $( \Lambda_{1} ,   \phi_{1} )$ and we identify this set with $\Gamma_{1}  \diagdown   \cY_{y_{1}}$ where, as above,
 \begin{eqnarray*}
\cY_{y_{1}}&=& \{ g  \in GL_{2}^{+}(\bQ)  \; | \; g  \rho_{0}  \in M_{2}(\widehat{ \bZ }) \}
 \end{eqnarray*}
and $\Gamma_{1} = SL_{2}(\bZ)$.  We will view this set as a subset of $\cL_{\Lambda_{1}}$ (which we have in turn identified  with $ \Gamma_{1}   \diagdown   GL_{2}^{+}(\bQ)$).

 \vspace{0.3 cm} \noindent  We define the hyperdistance at the level of $\bQ$-lattices, $\delta_{1}^{\bQ}$, as the restriction of $\delta_{1}$ to $\cP \cL_{y_{1}}^{\bQ} \times \cP \cL_{y_{1}}^{\bQ}$ and take
 \begin{eqnarray*}
\mathrm{d }_{1}^{\bQ} &=& \log \delta_{1}^{\bQ}\, .
 \end{eqnarray*}
 
 \vspace{0.3 cm} \noindent  These definitions coincide with the previous ones when the base $\bQ$-lattice is invertible in which case $\cL_{( \Lambda_{1} ,   \phi_{1} ) }^{\bQ} =\cL_{\Lambda_{1} }$. 
 
  \vspace{0.3 cm} 
  
 
\subsection{Representation of the $\GL_{2}$-system on the big picture} We use the identification $\cL_{( \Lambda_{1} ,   \phi_{1} ) }^{\bQ} =\cL_{\Lambda_{1} }$ in the case of an invertible $\bQ$-lattice to prove the following 

\begin{theorem}
\label{theorem1}
Let $( \Lambda_{1} ,   \phi_{1} )$ be an invertible $\bQ$-lattice represented up to scale by $y_{1}  = (\rho_{1}, z_{1}) \in M_{2}(\widehat{ \bZ } ) \times \bH$. Denote by $\cV = \cP \cL_{\Lambda_{1} }$ the set of vertices of the big picture given by projective classes of lattices commensurable with $\Lambda_{1} $ and by $\mathrm{d}_{1}$ its metric as defined in Lemma~\ref{hypermetric}. Then:
\begin{enumerate}
  \item The $C^{*}$-algebra $\cA$ admits a representation $\pi$ as an algebra of bounded operators acting on $\cH  = \ell^{2}(\cV)$.
  \item The time evolution $\sigma_{t}$ is implemented in this representation by the operator
 \begin{eqnarray*}
(\mathrm{H}  \xi) (\nu)& =&  \mathrm{d }_{1} ( \nu_{1} ,  \nu  )  \, \xi (\nu)
 \end{eqnarray*}
via
 \begin{eqnarray*}
\pi (\sigma_{t} f )   &=& e^{i t \mathrm{H} }  \pi (f)  e^{-i t \mathrm{H}}
\end{eqnarray*}
where $\nu_{1}$ is the vertex corresponding to $ \Lambda_{1}$.
\end{enumerate}
\end{theorem}

\begin{proof}

Let $( \Lambda_{1} ,   \phi_{1} )$ be an invertible $\bQ$-lattice represented up to scale by $y_{1}  = (\rho_{1}, z_{1}) \in M_{2}(\widehat{ \bZ } ) \times \bH$. For each element $\nu$ in $\cP \cL_{\Lambda_{1} } = \cV$ we can choose a representative given by a sublattice $\Lambda \subset  \Lambda_{1}$. At the matrix level this is equivalent to choosing the matrix representatives given by the map $s$ in Equation (\ref{eqsection}). Since $( \Lambda_{1} ,   \phi_{1} )$  is invertible we have that $\cP \cL_{y_{1}}^{\bQ}  = \cP \cL_{\Lambda_{1} } = \cV$ and we obtain the following identification for 
the set of $\bQ$-lattices commensurable with $( \Lambda_{1} ,   \phi_{1} )$:
 \begin{eqnarray*}
\cY_{y_{1}} &=&  \Gamma_{1}   \diagdown   M_{2}^{+}( \bZ  ) 
\end{eqnarray*}
Under this identification $ \delta_{1}$ is given by $\delta(g) = \det g$ for $g \in \cY_{y_{1}}$. We can then identify the Hilbert space $\cH$ with $\cH_{y_{1}}$ and obtain $\pi$ as the representation $\pi_{y_{1}}$ in (\ref{rep1}) of Section~\ref{section3}. As shown in \cite{ConnesMarcolli1}, in this representation the time evolution is implemented by the operator
 \begin{eqnarray*}
(\mathrm{H}_{y_{1}} \xi) (g)& =&  \log \det (g)  \xi (g)
 \end{eqnarray*}
which, under the above identifications, is given by $\mathrm{H}$.
\end{proof}

\vspace{0.3 cm} \noindent The spectrum of the operator $\mathrm{H}$ is given by:
 \begin{eqnarray*}
\Spec ( \mathrm{H} ) & =& \{  \mathrm{d }_{1} ( \nu_{0} ,  \nu  ) \, | \,   \nu \in \cV \} \\
&=&  \{ \log \det g \, | \,  g \in M_{2}^{+}( \bZ  )  \}
 \end{eqnarray*}
and the set $\cV_{n}\subset \cV$ corresponding to projective classes of lattices at hyperdistance $N$ to $ \Lambda_{1}$ is a basis for the eigenspace with eigenvalue $\log N$. If we denote by $P_{N}$ the projection onto this eigenspace we can view the $N$-th Hecke correspondence as an operator on $\cH$:
\begin{eqnarray*}
(T_{N} \xi )(  \nu )  &=& \sum_{ \delta_{1}( \nu ,  \nu'  ) = N } \xi (  \nu' )
 \end{eqnarray*}

\vspace{0.3 cm} \noindent In a similar way the various finite subgraphs whose pointwise and setwise stabilizers are of interest single out subespaces of $\cH$ which are invariant under the corresponding congruence subgroups. We will be interested in particular in the spaces corresponding to threads and snakes: 
\begin{eqnarray*}
\cH_{\mathfrak{t}_{N}} = Span\{ \nu \, |\, \nu \in {\mathfrak{t}_{N}} \} &,&\cH_{\mathfrak{s}_{N}} = Span\{ \nu \, |\, \nu \in {\mathfrak{s}_{N}}  \}
 \end{eqnarray*}
and will denote by $P_{\mathfrak{t}_{N}}$ and  $P_{\mathfrak{s}_{N}}$ the projections onto these subspaces. By the results discussed at the end of Section~\ref{Section2} the restriction to $\overline{\Gamma}_{0}(N)$ of the $\PGL_{2}(\bQ)$-action on $\cH$ is trivial on the subspaces $\cH_{\mathfrak{t}_{N}}$ and $\cH_{\mathfrak{s}_{N}}$ and the $\overline{\Gamma}_{0}(N)$-cosets given by the Atkin-Lehner involutions act as symmetries of these spaces.


\vspace{0.3 cm} \noindent
\subsection{Compatibility with the Bost-Connes  $\GL_{1}$-system}

\vspace{0.3 cm} \noindent The one-dimensional analog of the $\GL_{2}$-system, defined in terms of commensurability classes of one dimensional $\bQ$-lattices, plays a central role in the theory. This system studied by Bost and Connes in \cite{BostConnes} encodes the arithmetic properties of abelian extensions of $\bQ$ (i.e. cyclotomic extensions and the Kronecker-Weber theorem) and was the first quantum statistical mechanical system of this sort to be studied. This system was originally defined in terms of the Hecke algebra of the pair $(P_{\bQ}, P_{\bZ})$ where $P$ denotes the ``$ax + b$" group given for a commutative ring $R$ by the group of matrices  
 \begin{eqnarray*}
P_{R} &=& \{   \left(\begin{array}{cc} 1 & b \\ 0 & a \end{array}\right)  \, |  \, a \in R^{*}, \, b \in R  \}.
 \end{eqnarray*}

\vspace{0.3 cm} \noindent The pair $(P_{\bQ}, P_{\bZ})$ is almost normal in the sense that the orbits of $P_{\bZ}$ acting on the left on the quotient $P_{\bQ}/ P_{\bZ}$ are finite and it is possible to associate to it a $C^{*}$-algebra $C^{*}_{r}(P_{\bQ}, P_{\bZ})$ which comes equipped with a natural time evolution  (cf. \cite{BostConnes}  for the various definitions). This is the Bost-Connes system also referred as the $\GL_{1}$-system.

\vspace{0.3 cm} \noindent  It turns out (cf. loc. cit) that for each prime number $p$ the local quotients $P_{ \bQ_{p} } / P_{ \bZ_{p} }$ are isomorphic to the Bruhat-Tits building of $\GL_{2}(\bQ_{p})$ and the $C^{*}$-algebra of the $\GL_{1}$-system can be represented on the space of $\ell^{2}$ functions on the set of vertices of the restricted product of these buildings
\begin{eqnarray*}
\Delta &=&  \prod  P_{ \bQ_{p} } / P_{ \bZ_{p} }  \, = \,  P_{ \bA_{f} } / P_{ \widehat{\bZ }}   \, = \,  P_{ \bQ} / P_{ {\bZ }} 
 \end{eqnarray*}

\vspace{0.3 cm} \noindent Taking into account Corollary~\ref{collorary1} we can identify $\Delta$ with the set $\cV$ of vertices of the big picture obtaining a representation of the $\GL_{1}$-system as an algebra of bounded operators acting on $\cH= \ell^{2}(\cV)$.  By Theorem~\ref{theorem1} the $\GL_{2}$-system also acts on the space $\cH= \ell^{2}(\cV)$ and the relation between the two actions is implemented by the correspondence between the two systems defined in \cite{CMR}. This correspondence and its implementation on the big picture essentially account for the role played by the determinant $ \det: \GL_{2} \rightarrow \GL_{1}$ on Shimura's theory relating the arithmetic of the the field of modular functions of level $N$ to that of the cyclotomic field $\bQ(\e ^{2\pi i \frac{1}{N}})$ (cf. \cite{Shimura}). 

\vspace{0.3 cm} 


\section{Arithmetic of groups like $\Gamma_{0}(N)$ and the $\GL_{2}$-system}

 \noindent In this section we explore some aspects of moonshine theory related to congruence subgroups of moonshine type.  As seen above, the main idea behind Conway's big picture is to understand these and similar groups in terms of their action on sets of commensurable lattices. Commensurable lattices up to homothety correspond to isogenous elliptic curves and passing to $\bQ$-lattices adds a labeling of torsion points on these curves. If we view automorphic functions corresponding to congruence subgroups as functions on sets of lattices then passing to $\bQ$-lattices leads to considering the adelic counterpart of these groups and the corresponding Galois action on automorphic functions as encoded in terms of this torsion data.
\vspace{0.3 cm} \noindent

\subsection{The modular field and groups of moonshine type}
\label{moonshinetype}

We start this section by recalling some basic facts about the arithmetic theory of automorphic functions, these will provide a natural bridge between moonshine theory and the arithmetic properties of the $GL_{2}$-system. We refer the reader to \cite{Shimura} for an extended treatment of the results used in this section.

\vspace{0.3 cm} \noindent For a positive integer $N$ we denote by $\cF_{N}$ the field of $\Gamma (N)$-automorphic functions with Fourier coefficients in $\bQ(\zeta_{N})$, $\zeta_{N}$ a primitive $N$-th root of unity. The \emph{modular field} is the limit
 \begin{eqnarray*}
\cF  &=& \lim_{\overrightarrow{n}} \,  \cF_{N}.
\end{eqnarray*}
There is a natural homomorphism
 \begin{eqnarray*}
\tau \, : \, GL_{2}(\bA_{f})  &\longrightarrow & \Aut( \cF)
\end{eqnarray*}
where we denote by $\bA_{f} $ the ring of finite adeles of $\bQ$. To see this recall that the field $\cF$ can be generated  the by Fricke functions $\{f_{a} \ | \,   a\in \bQ^{2}/ \bZ^{2} \}$ where $f_{a}(z)$ is defined in terms of the coordinates of the torsion point corresponding to $a$ in the elliptic curve with period lattice $\bZ+z\Z$ (cf.  \cite{Shimura}). Taking $GL_{2}(\bA_{f}) =  GL_{2}(\widehat{\bZ}) GL_{2}^{+}(\bQ)$ via the embedding of $GL_{2}^{+}(\bQ)$ in $GL_{2}(\bA_{f})$ the above morphism corresponds to:
 \begin{eqnarray*}
\tau(u) \, :\, f_{a} & \mapsto &  f_{ u a} \qquad \text{ for }u\in  GL_{2}(\widehat{\bZ}) \\
\tau(g) \, :\, f_{a} & \mapsto &  f_{ a} \circ g \quad  \text{ for }g\in  GL_{2}^{+}(\bQ)
\end{eqnarray*}
(these are equivalent in $GL_{2}(\widehat{\bZ}) \cap GL_{2}^{+}(\bQ) = SL_{2}(\bZ)$). The homomorphism $\tau$ fits into an exact sequence:
 \begin{eqnarray*}
1 \; \longrightarrow  \;   \bQ^{*}   \;  \longrightarrow    \; GL_{2}(\bA_{f})  \;    \overset {\tau}{\longrightarrow}   \;  \Aut( \cF)   \;  \longrightarrow   \;  1 .
\end{eqnarray*}

\vspace{0.3 cm} \noindent We say that a point $z_{1}\in \mathfrak{H}$ is \emph{generic} if the value $j(z_{1})$ of the $j$-invariant of the elliptic curve $\bC / (z_{1} \bZ + \bZ)$ is transcendental. Given $z_{1}\in  \mathfrak{H}$, generic evaluation at $z_{1}$ defines an isomorphism
\begin{eqnarray*}
\cF \longrightarrow  \cF_{z_{1}} = \{ f(z_{1}) \, | \, f \in \cF \}
\end{eqnarray*}
We call the  subfield $\cF_{z_{1}} \subset \bC$ the specialization of $\cF$ at $z_{1}$ and denote the corresponding Galois action of $GL_{2}(\bA_{f}) /\bQ^{*}$ on $\cF_{z_{1}}$ by $\tau_{z_{1}}$. We say that lattice $\Lambda_{1}$ (resp. a $\bQ$-lattice $(\Lambda_{1}, \phi_{1})$) is generic if it is represented up to scale by  $z_{1}\in \mathfrak{H}$ (resp. $(\rho_{1}, z_{1}) \in M_{2}(\widehat{ \bZ } ) \times \mathfrak{H}$) with  $z_{1}$ a generic point.

\begin{proposition}[cf. \cite{Shimura}]
\label{propshi}
Let $S$ be an open subgroup of $GL_{2}(\bA_{f})$ with $\bQ^{*} \subset S$ such that  $S / \bQ^{*}$ is compact. Then
\begin{enumerate}
  \item The fixed field $\cF_{S} = \cF^{\tau(S)}$ is finitely generated over $\bQ$ and $\cF$ is a Galois extension of $\cF_{S}$ with Galois group $\Gal (\cF / \cF_{S}) = \tau(S)$.
  \item The group $\bar{\Gamma}_{S} = (  GL_{2}^{+}(\bQ)  \cap S) / \bQ^{*}$ is commensurable with $\bar{\Gamma}_{1} = PSL_{2}(\bZ)$.
  \item The field $\bC \cF_{S}$ coincides with the field of all $\bar{\Gamma}_{S}$-automorphic functions on $\mathfrak{H}$.
\end{enumerate}
\end{proposition}

\vspace{0.3 cm} \noindent As above we identify elements of $PGL_{2}^{+}(\bQ)$ with the corresponding fractional linear transformations on $\mathfrak{H}$. In what follows we will restrict our attention to subgroups $\Gamma \subset GL_{2}^{+}(\bQ)$ such that $\bar{\Gamma}$ satisfies the following properties:
\begin{enumerate}
  \item There exist an $N$ such that $\bar{\Gamma}$ contains $\bar{\Gamma}_{0}(N)$ with finite index.
  \item 
  \label{cond2} The map $z \mapsto z + k $ belongs to $\bar{\Gamma}$ if and only if $k \in \bZ$.
  \item $\bar{\Gamma}$ is a genus zero congruence group, that is $X_{\Gamma} = \widetilde{\mathfrak{H} } / \bar{\Gamma} $ is a compact Riemann surface of genus zero where $ \widetilde{\mathfrak{H} } = \mathfrak{H} \cup \bQ \cup \{ i \infty \}$.
  
\end{enumerate}
We will call such a group a \emph{congruence subgroup of moonshine type}\footnote{This definition is local to this article. In \cite{ConwayMcKaySebbar} Conway, McKay and Sebbar intrinsically characterized the $171$ groups occurring in Monstrous Moonshine, these groups fall into our definition of ``moonshine-type" and it might be convenient at a later stage to restrict our attention to them.}. Having genus zero we have that $ X_{\Gamma} \simeq \bP ^{1}$ and the field of $\Gamma$-automorphic functions can be generated by one element $f_{\Gamma}$, we will call such generator a \emph{principal modulus} of $\Gamma$. Because of condition \ref{cond2} above the principal modulus $f_{\Gamma}$ admits a Fourier expansion in powers of $q=\e ^{2 \pi i z}$. As mentioned in Section~\ref{section1} character values of the Monster group $\bM$ give rise to generating series which by the monstrous moonshine theorem coincide with the Fourier series of principal moduli of congruence subgroups satisfying the above properties.

\vspace{0.3 cm} \noindent The following converse to Proposition~\ref{propshi} will be central to our study of congruence subgroup of moonshine-type and their principal moduli in relation to the $GL_{2}$-system:

\begin{proposition}
\label{SuperShimura}
Let $\Gamma$ be a congruence subgroup. Then there exist an open subgroup $S_{\Gamma}$ of $GL_{2}(\bA_{f})$ with $\bQ^{*} \subset S_{\Gamma}$ and $S_{\Gamma} / \bQ^{*}$ compact such that the group $\bar{\Gamma}$ is isomorphic to $ (GL_{2}^{+}(\bQ)  \cap S_{\Gamma}) / \bQ^{*}$. If $\Gamma$ is of moonshine-type and $f_{\Gamma}$ is a principal modulus for $\Gamma$ with rational Fourier coefficients then fixed field of $\tau(S_{\Gamma})$ is given by $\cF_{S_{\Gamma}} = \bQ (f_{\Gamma})$.
\end{proposition}
\noindent Note that in this case the field $\bC \cF_{S_{\Gamma}}$ coincides with $\bC(f_{\Gamma})$, the function field of $X_{\Gamma}$.


\subsection{Arithmetic properties of the $GL_{2}$-system}

\vspace{0.3 cm} \noindent In the context of quantum statistical mechanical systems the counterparts of special values occurring in analytic number theory are given by expectation values of arithmetic observables at equilibrium states of the system. Relevant Galois groups arise then as groups of symmetries of the system whose action on equilibrium states commutes with the action on these special values.

\vspace{0.3 cm} \noindent The notion of equilibrium states of a quantum statistical mechanical system is given in terms of the Kubo-Martin-Schwinger condition which we now recall. A \emph{state} $\varphi$ on a $C^{*}$-algebra $\cA$ is by definition a positive linear functional of norm one. Given a quantum statistical mechanical system $(\cA, \sigma_{t})$, a state $\varphi$ on $\cA$ satisfies the Kubo-Martin-Schwinger condition at inverse temperature $\beta \in (0 , \infty)$ if for every $a, b \in \cA$ there exist a continuous function:
\begin{eqnarray*}
f : \{ z\in \mathbb{C} \mid 0 \leq \Im(z) \leq \beta \} \longrightarrow \bC,
\end{eqnarray*}
holomorphic in the interior of the strip, and such that
\begin{eqnarray*}
f(t) = \varphi (a \sigma_t (b)) &\text{and} & f(t + i \beta) = \varphi (\sigma_t (b) a)
\end{eqnarray*}
for all $t\in \bR$.  We call such state a $KMS_{\beta}$ state. A $KMS_\infty$ state is by definition a weak limit of $KMS_{\beta}$ states as $\beta \rightarrow \infty$. For all $\beta$ the set of $KMS_{\beta}$ states  forms a Choquet simplex. We denote the set of extremal $KMS_{\beta}$ states by $\cE_{\beta}$. We refer the reader to \cite{Bratteli} for various generalities about quantum statistical mechanical systems and their $KMS$-states.

\vspace{0.3 cm} \noindent In the case of the $GL_{2}$-system the structure of $KMS$-states is well understood. The  $GL_{2}$-system exhibits symmetry breaking at inverse temperatures $\beta=1$ and $\beta=2$ having no $KMS_{\beta}$ states in the range $0< \beta < 1$ and a single $KMS_{\beta}$ in the range $1 \leq \beta \leq 2 $ while for low temperatures $ 2 < \beta $ the set of extremal states can be identified with the set of invertible $\bQ$-lattices and admits a transitive action of  $GL_{2}(\widehat{\bZ}) $ (cf. \cite{ConnesMarcolli1, LacaEtAl1}). Given an invertible $\bQ$-lattice  $(\Lambda_{1},\phi_{1})$ represented up to scale by $y_{1} = (\rho_{1}, z_{1})$ we will denote by $\varphi_{y_{1}} \in \cE_{\infty}$ the associated extremal $KMS$ state at zero temperature.

\vspace{0.3 cm} \noindent The arithmetic of the  $GL_{2}$-system is encoded via an \emph{arithmetic sub-algebra of unbounded multipliers} $\cA_{\bQ}$ which is given by functions on 
 \begin{eqnarray*}
\cY &=& \{  (g,  \rho, z) \in GL_{2}^{+}(\bQ) \times M_{2}(\widehat{ \bZ }) \times \mathfrak{H} \; | \; g  \rho \in M_{2}(\widehat{ \bZ }) \}
\end{eqnarray*}
satisfying various finiteness conditions. More precisely:
\begin{definition} A $\bC$-valued continuous function on $\cY$ such that
\begin{eqnarray*}
f(  \gamma g,  \rho, z) = f(g, \rho,  z) &\text{   and    }&  f(g  \gamma ,  \rho, z) = f(g, \gamma  \rho,  \gamma z)\, ,\; \forall \gamma  \in \Gamma_{1} =\SL_{2}(\bZ)
\end{eqnarray*}
is \emph{arithmetic} if it satisfies the following properties:
\begin{itemize}
  \item $f$ has finite support on $g$,
  \item $f$ has finite level in $\rho$,
  \item for all $(g,  \rho) \in GL_{2}^{+}(\bQ) \times M_{2}(\widehat{ \bZ })$ the function $z \mapsto f(  g,  \rho, z)$ belongs to $\cF$,
\item for any $u\in  GL_{2}(\widehat{\bZ})$ we have $ f(  g,  u \rho, z) = \tau(u)  f(  g',  \rho, z)$ where $gu = u'g' \in GL_{2}(\widehat{\bZ})  GL_{2}^{+}(\bQ)$.
\end{itemize}
\end{definition}
\vspace{0.3 cm} \noindent  The set of all arithmetic functions becomes a $*$-algebra $\cA_{\bQ}$ whose product and involution are given as for elements of $\cA$. The product of an arithmetic element with an element of $\cA$ belongs to $\cA$ and we can view the algebra $\cA_{\bQ}$ as a subalgebra of the algebra of unbounded multipliers of $\cA$. In particular it is possible to evaluate states of the algebra $\cA$ at elements of $\cA_{\bQ}$.

\vspace{0.3 cm} \noindent  As shown in \cite{ConnesMarcolli1} given a extremal KMS state at zero temperature $\varphi \in  \cE_{\infty}$ corresponding to a generic invertible $\bQ$-lattice represented up to scale by $(\rho_{1}, z_{1})$ the values $\varphi (\cA_{\bQ})$ generate the specialization of the modular field $\cF$ at $z_{1}$. The Galois action of $GL_{2}(\bA_{f}) $ is then implemented by an action by symmetries on the $GL_{2}$-system whose induced action on states commutes with the Galois action on $\cF$ once passing to this specialization.

\vspace{0.3 cm}


\subsection{Subgroups of moonshine type and symmetries of the $\GL_{2}$-system} $\,$

\vspace{0.3 cm} \noindent
The group $GL_{2}(\bA_{f}) $ acts as a group of symmetries of the $GL_{2}$-system and this action induces an action on its KMS states.  Given a congruence subgroup $\Gamma$ we will be interested in the action of $S_{\Gamma} \subset GL_{2}(\bA_{f})$ both on the system and its KMS states. In what follows we briefly recall some properties of this adelic action referring to \cite{ConnesMarcolli1} for details, after this we restrict to the action of the above adelic counterparts of moonshine groups.

\vspace{0.3 cm} \noindent
The key step in understanding the symmetries of the $GL_{2}$-system comes from the fact that the algebra of observables $\cA$ can be realized as a ``full corner" of an algebra of $\Gamma_{1}\times\Gamma_{1}$ invariant functions on the space
\begin{eqnarray*}
\cX &=& GL_{2}^{+}(\bQ) \times M_{2}(\bA_{f}) \times \mathfrak{H} 
\end{eqnarray*}
The full corner giving the algebra of the $GL_{2}$-system corresponds to the projection on the space of functions with support on $\cY$. The action of $GL_{2}(\bA_{f}) $ is then implemented via multiplication in the second factor of the above product. Because of the required compatibility with the projection defining $\cA$ the action of $GL_{2}(\bA_{f})$ will correspond then to an action by algebra automorphisms of $GL_{2}(\widehat{\bZ}) $ together with an action by algebra endomorphisms of $M_{2}^{+}(\bZ)$.

\vspace{0.3 cm} \noindent  In the same manner we can define a natural action $\vartheta$ of thegroup $GL_{2}(\bA_{f}) / \bQ^{*}$ on the arithmetic sub-algebra of unbounded multipliers $\cA_{\bQ}$. As mentioned above this action will be intertwined with the Galois action on the specialization of $\cF$ corresponding to a generic invertible $\bQ$-lattice via the $GL_{2}(\bA_{f})$-action on states.

\vspace{0.3 cm} \noindent Let $\Gamma$ be a congruence subgroup and let $S_{\Gamma}$  be the corresponding open subgroup of $GL_{2}(\bA_{f})$ as in Proposition~\ref{SuperShimura}. We will refer elements in
\begin{eqnarray*}
\cA_{\bQ}^{\Gamma} &=& \{ f \in \cA_{\bQ} \, | \,  \vartheta_{u} f = f \text{ for all }   u \in S_{\Gamma} \}
\end{eqnarray*}
as $\Gamma$-\emph{automorphic arithmetic multipliers} of the system $(\cA, \sigma_{t})$. We have the following:

\begin{theorem}
Let $\Gamma$ be a congruence subgroup of moonshine type with principal modulus $f_{\Gamma}$. Then for any invertible generic $\bQ$-lattice  $(\Lambda_{1},\phi_{1})$ represented up to scale by $y_{1}=  (\rho_{1}, z_{1})$ with associated KMS-state $\varphi_{y_{1}} \in  \cE_{\infty}$ we have
\begin{enumerate}
  \item the values $\varphi(\cA_{\bQ}^{\Gamma})$ generate the field $\cF_{\Gamma,z_{1}} = \bQ(f_{\Gamma}(z_{1}))$ over $\bQ$,
  \item  for all $u \in S_{\Gamma} $ the element $\tau_{z_{1}}( \rho_{1} u \rho_{1}^{-1} )$ acts as the identity in $\cF_{\Gamma,z_{1}}$.
\end{enumerate}
\end{theorem}

\vspace{0.3 cm} \noindent We can now return to the results of Section~\ref{section4} and study the relation between the $\PGL_{2}(\bQ)$-action on $\cH = \ell^{2}(\cV)$ and the symmetries of the $\GL_{2}$-system. This relationship is encoded in the following

\begin{proposition}
Let $\Gamma =  \Gamma_{0}(N)$ and let $S_{\Gamma}$ be the corresponding open subgroup of $GL_{2}(\bA_{f})$ viewed as a group of symmetries of the $\GL_{2}$-system. Then for any element $a$ of the  algebra of $\Gamma$-automorphic arithmetic multipliers $\cA_{\bQ}^{\Gamma}$ and any positive $N$ we have 
\begin{eqnarray*}
a P_{\mathfrak{t}_{N}} = P_{\mathfrak{t}_{N}} a  &, & a P_{\mathfrak{s}_{N}} = P_{\mathfrak{s}_{N}} a  
\end{eqnarray*}
where $P_{\mathfrak{t}_{N}}$ and $P_{\mathfrak{s}_{N}}$ are the projections onto the subspaces of $\cH=\ell^{2}(\cV)$ corresponding to the $(N|1)$-thread and $(N|1)$-snake. 
\end{proposition}


\vspace{0.3 cm} \noindent
\section{Replicability and the arithmetic subalgebra}
\vspace{0.3 cm} \noindent

Replicability is the operation on series as those in (\ref{hauptmoduli1}) which reflects the power map structure on the monster group $\bM$. This is made precise in the following way: let $f(z)$ be a function on the upper half plane given by a power series of the form
\begin{eqnarray}
\label{hauptmoduli3}
f(z)  &=& q^{-1} +  a_{1} q +  a_{2} q^{2} + \dots, \qquad q =\e^ {2\pi i z}
\end{eqnarray}
with $ a_{i} \in \bQ$. For a positive integer $k$ the \emph{$k$-th replicate power} of $f$, denoted by $f^{(k)}$, is defined inductively by the equation
\begin{eqnarray*}
f^{(k)}(k z)  &=& - \sum_{ad=k, a\neq k,  0\leq b < d} f^{(a)} \left( \frac{az + b}{d}  \right) \; + \; Q_{k}(f)
\end{eqnarray*}
where $Q_{k} \in \bQ[x]$ is the unique polynomial such that $Q_{k}(f(z))$ has the form 
\begin{eqnarray*}
Q_{k}(f(z))  &=& q^{-k}  +  b_{1}q +  b_{2}q^{2} + \dots
\end{eqnarray*}

\vspace{0.3 cm} \noindent  It was conjectured in \cite{ConwayNorton} and showed in \cite{Borcherds} that if $f(z) = f_{\langle m \rangle}(z) $ is the principal modulus corresponding to the conjugacy class of an element $m \in \bM$ then $f^{(k)}(z)$ is the principal modulus corresponding to the conjugacy class of its $k$-th power $m^{k} \in \bM$. 

\begin{definition}
A function  $f(z)$  on the upper half plane admitting a power series expansion with rational coefficients of the form (\ref{hauptmoduli3}), is said to be replicable if for every $k$ the $k$-th replicate power $f^{(k)}$ is given by a power series of the same form (\ref{hauptmoduli3}) with rational coefficients. 
\end{definition}

\vspace{0.3 cm} \noindent Note that in principle the series expansion of $f^{(k)}$ could have fractional powers of $q$ and its coefficients could be irrational so this condition is non vacuous. It is known moreover that replicability implies strong arithmetic constraints on $f$. For an account on replicable functions and some of their properties the reader might consult \cite{McKaySebbar} and references therein.   

\vspace{0.3 cm} \noindent  Now let $\Gamma$ be a congruence subgroup of moonshine-type with principal modulus $f_{\Gamma}$ having rational Fourier coefficients. It has been shown by Cummins and Norton in \cite{CumminsNorton} that in this case
$f_{\Gamma}$ is replicable. The proof relies heavily on the structure of $\Aut(\cF)$. Given $N$ such that $\bar{\Gamma}_{0}(N)\subset \bar{\Gamma}$ the principal modulus $f_{\Gamma}$ belongs to $\cF_{N}$ which is generated by the Fricke functions $f_{a}$ with $ a \in \left( \frac{1}{N} \bZ \right )^{2}/ \bZ^{2} $. Cummins and Norton prove inductively that the $k$-th replicate $f_{\Gamma}^{(k)}$ of $f_{\Gamma}$ admits a Fourier expansion in powers of $q=\e ^{2 \pi i z}$ by noting that $f_{\Gamma}^{(k)} \in \cF_{M}$ for some $M$ and identifying a congruence subgroup $\Gamma'$ of $\GL_{2}^{+}(\bQ)$ for which $f_{\Gamma}^{(k)}$ is $\Gamma'$-automorphic. Using this fact together with Galois theory arguments it is then shown that the series defining $f_{\Gamma}^{(k)}$, which is a priori a series admitting fractional powers of $q$ and cyclotomic coefficients, is indeed a series in $q$ with rational coefficients.

\vspace{0.3 cm} \noindent In view of the previous sections it is natural to expect that $f_{\Gamma}$ and its replicates admit counterparts in $\cA_{\bQ}$. This is indeed the case as can be seen for instance by expressing everything in terms of Fricke functions and using the corresponding generalizations of these to elements of $\cA_{\bQ}$ as defined in \cite{ConnesMarcolli1}. Likewise it is possible to follow up the proof of \cite{CumminsNorton} in this context using the various operations given by the Galois action in terms of which replicability is defined. At a conceptual level however it would be better to start by expressing replicability in terms of operators in $\cA_{\bQ}$ and then study the effect of these in $\Gamma$-automorphic arithmetic multipliers. Incorporating replicability in the framework of $\bQ$-lattices in this way might clarify some of its arithmetic aspects.

\vspace{0.3 cm} \noindent For the case of the (normalized) $j$-function $J(z) = f_{\langle id_{\bM} \rangle}(z) =  f_{\SL_{2}(\bZ)}(z) $ one has $J^{(k)}(z) =J(z)$ for all $k$ and the values of the sum 
\begin{eqnarray}
 \sum_{ad=k,  0\leq b < d} f^{(a)} \left( \frac{az + b}{d}  \right)
\end{eqnarray}
is that of the $k$-th Hecke operator acting on $J$. For an arbitrary replicable function $f$ we consider this sum as a \emph{generalized Hecke operator}. A systematic study of these operators is yet to be carried out. Such study should incorporate the Galois theoretic aspects mentioned above as well as connections with Adams operations and lambda rings which become apparent when looking at the effect of replicability on power series from a purely formal level.

\vspace{0.3 cm} 
\section{Concluding remarks}

 \vspace{0.3 cm} \noindent It seems plausible that generalized Hecke operators might give rise to \emph{generalized modular Hecke algebras} similar to the modular Hecke algebras introduced in \cite{ConnesMoscovici} which are algebras obtained from the action of Hecke operators on modular forms.  In the case of rational principal moduli one would expect the Hecke algebras of the corresponding congruence subgroups to appear as subalgebras of these generalized modular Hecke algebras. 

 \vspace{0.3 cm} \noindent For the case of $\Gamma_{1} =SL_{2}(\bZ)$ the corresponding modular Hecke algebra furnishes the ``holomorphic part" of the algebra of observables of the $GL_{2}$-system. Modular Hecke algebras admit an action of a Hopf algebra $\cH_{1}$ which arises  in the study of codimension $1$ foliations. At the level of cohomology the cyclic cohomology of the algebra $\cH_{1}$ can be given in terms of the Lie algebra cohomology of the algebra of formal vector fields. Pushing analogies one could look for an action of Borcherds' monster Lie algebra on the generalized modular Hecke algebras corresponding to congruence subgroups of moonshine type. Here one would expect such an action to be given in terms of a suitable Hopf algebra. Looking at the results in \cite{JurisichEtAl} in this light could provide the right hints in this direction. Direct links with the theory of vertex algebras are of course desirable, however we have not yet explored possible connections.

\vspace{0.3 cm} 
\noindent Group actions on trees and buildings lead in a natural way to $C^{*}$-algebras reflecting the arithmetic and combinatoric properties of these actions (cf. \cite{BaumgartnerEtAl}, \cite{CornelissenEtAl}). As a different approach to some of the ideas above it would be worth studying the $C^{*}$-algebras resulting from the action of congruence subgroups of moonshine type on the big picture.

\vspace{0.3 cm} \noindent As a first simple instance we can study the finite dimensional case corresponding to the $C^{*}$-algebra of the $\bar{\Gamma}$-invariant tree associated to a congruence subgroup of moonshine type $\bar{\Gamma} \subset PGL_{2}^{+}(\bQ)$. Given a subgroup $\bar{\Gamma}$ of $PGL_{2}^{+}(\bQ)$ commensurable with $\bar{\Gamma}_{1}= PSL_{2}(\bZ)$ the $\bar{\Gamma}$-action on the big picture has finite orbits. As observed in \cite{Conway} we can project the points of the orbit $\bar{\Gamma}\nu_{1}$ (with $\nu_{1}$ the vertex corresponding to $ \Lambda_{1} $ in $\cV = \cP \cL_{\Lambda_{1} }$) onto each of the $p$-adic trees of the big picture and adjoin to $\bar{\Gamma}\nu_{1}$ the vertices of the paths connecting the corresponding images. The resulting graph will be a finite $\bar{\Gamma}$-invariant tree $\cT_{\bar{\Gamma}}$ which becomes a directed graph by the choice of $\nu_{1}$ as a root.  The corresponding  $C^{*}$-algebra is a matrix algebra generated by partial isometries corresponding to the vertices of $\cT_{\bar{\Gamma}}$ and projections corresponding to its edges (cf. \cite{Kumjian}). Already at this level one observes interesting phenomena as for example the presence of roots of unity coming from the natural circle action on these algebras.


 \bibliographystyle{amsalpha}

\end{document}